\documentclass[letterpaper,twocolumn]{article} 
\usepackage{times}

\usepackage{graphicx}
\usepackage{latexsym}

\usepackage{amsthm}
\usepackage{amssymb}	
\usepackage{multirow}
\usepackage{pgfplots}
\usepackage{tikz}
\usetikzlibrary{shapes,arrows}
\usetikzlibrary{positioning}

\pgfplotsset{width=7cm,compat=1.8}
\usepackage{pgfplotstable}

\usepackage{helvet}  
\usepackage{courier}
\usepackage[hyphens]{url} 
\usepackage{graphicx}
\usepackage{caption} 
\usepackage{algorithm}
\usepackage{algorithmic}
\usepackage{newfloat}
\usepackage{listings}
\usepackage{geometry}
\floatstyle{ruled}
\newfloat{listing}{tb}{lst}{}
\floatname{listing}{Listing}
\usepackage[utf8]{inputenc}
\usepackage[english]{babel}
\usepackage{enumerate}
\usepackage{ amssymb}
\usepackage{amsmath}
\usepackage{amsthm}
\usepackage{ dsfont }
\usepackage{xspace}
\usepackage{subcaption}
\usepackage{etoolbox}
\newcommand{\calA}{{\cal A}}

\newcommand{\calC}{{\cal C}}

\newcommand{\calL}{{\cal L}}

\newcommand{\calS}{{\cal S}}

\newcommand{\vphi}{\varphi}

\newcommand{\Lstrict}{\calL_{<}^\calA}
\newcommand{\Lnonstrict}{\calL_{\le}^\calA}

\newcommand{\Gammadestrict}{\Gamma^{(\leq)}}

\newcommand\mybox[2][]{\tikz[overlay]\node[draw, inner sep=1.0pt, anchor=text, rectangle,#1] {#2};\phantom{#2}}

\title{Towards Fast Algorithms for the Preference Consistency Problem Based on Hierarchical Models (Extended Version of IJCAI'16 paper)}

\author{Anne-Marie George, Nic Wilson, Barry O'Sullivan \\
Insight Centre for Data Analytics, 
School of Computer Science and IT \\
University College Cork, Ireland \\
\{annemarie.george, nic.wilson, barry.osullivan\}@insight-centre.org
}

\begin{document}

\maketitle
	
\thispagestyle{empty}
\newtheorem{thm}{Theorem}[section]
\newtheorem{lem}[thm]{Lemma}
\newtheorem{cor}[thm]{Corollary}
\newtheorem{con}[thm]{Conjecture}
\newtheorem{prop}[thm]{Proposition}
\newtheorem{definition}[thm]{Definition}
\newtheorem*{ex}{Example}
\newcommand\HCLP{\mathit{HCLP}}
\newcommand\length{\mathit{lenght}}

\thispagestyle{empty}

\begin{abstract}
In this paper, we construct and compare algorithmic approaches to solve the Preference Consistency Problem for preference statements based on hierarchical models. 
Instances of this problem contain a set of preference statements that are direct comparisons (strict and non-strict) between some alternatives, and a set of evaluation functions by which all alternatives can be rated. 
An instance is consistent based on hierarchical preference models, if there exists an hierarchical model on the evaluation functions that induces an order relation on the alternatives by which all relations given by the preference statements are satisfied. 
Deciding if an instance is consistent is known to be NP-complete for hierarchical models. 

We develop three approaches to solve this decision problem. 
The first involves a Mixed Integer Linear Programming (MILP) formulation, 
the other two are recursive algorithms that are based on properties of the problem by which the search space can be pruned. 
Our experiments on synthetic data show that the recursive algorithms are faster than solving the MILP formulation and that the ratio between the running times increases extremely quickly.
\end{abstract}

\section{Introduction}
In many fields like recommender systems and multi-objective optimization, 
one wants to reason over user preferences. 
In these problems, it is often difficult or excessively time-consuming to elicit all user preferences. 
The Preference Deduction Problem (PDP) aims at eliciting only few preferences and inferring more preferences from the given ones; this might then be used in a conversational recommender system, 
for example, to help choose which items to show to the user next~\cite{BridgRicci07,Trabelsi11}. 
Typically, an assumption is made on the type of order relation that the user (implicitly) uses to express the preference statements. Such order relations can be, e.g., comparing alternatives by the values of the evaluation functions lexicographically~\cite{WilsonECAI14}, by Pareto order, weighted sums~\cite{MCDA2005}, based on hierarchical models~\cite{WilsonGOSIJCAI2015longer} or by conditional preferences structures as CP-nets~\cite{BoutilierBDHP04} and partial lexicographic preference trees~\cite{LiuT15}. 
Here, the choice of the order relation can lead to stronger or weaker inferences and can make solving PDP computationally more or less challenging. 
In a recommender system or in a multi-objective decision making scenario, 
the user should only be presented with a relatively small number of solutions, 
hence, a strong order relation is required. 
Using PDP based on a lexicographic models has been shown to be successful in reducing the number of solutions, 
however, computationally can be more expensive. 
See~\cite{marinescu2013moopt} and~\cite{GeorgeRW15} for comparisons between order relations in a multi-objective optimization framework. 
While PDP based on hierarchical models yields an even lower number of solutions, it is coNP-complete and computationally expensive.
The approach taken by the Preference Deduction Problem contrasts with learning a single model that coincides with the user preferences as in~\cite{Fuern-Hueller-Preference-Learning,Dombi-Learning-Lex-EJOR-2007,FlachM07-Lex-ranker,Huellermeier-PL-12-Learn-Lex-Trees,BoothCLMS10}.

In this paper, we concentrate on another problem that helps solving PDP: the Preference Consistency Problem (PCP).
This is the problem of deciding whether given user preferences are consistent, 
i.e., not contradicting each other. 
In terms of hierarchical preference models, PCP determines whether there exists a hierarchical model of evaluation functions such that the induced order relation on the alternatives satisfies all preference statements given by the user.  
For hierarchical preference models, PCP is NP-complete and PCP and PDP are mutually expressible~\cite{WilsonGOSIJCAI2015longer}, i.e., PCP can be solved directly by solving PDP and vice versa. 
The main issue in this paper is to find fast algorithms to solve PCP (and thus PDP) for hierarchical models by exploiting the problem's structure and to compare their running times on a set of synthetic data.

The next section gives an  introduction to hierarchical preference models, 
their induced order relation on alternatives, and the Preference Consistency Problem. 
Section~\ref{sec: MILP} gives a Mixed Integer Linear Programming formulation for PCP that will function as a baseline for our runtime comparisons. 
Section~\ref{sec: rec algo} discusses the exploitation of properties of PCP instances yielding two recursive algorithms. 
In Section~\ref{sec: experiments}, the conducted experiments and their outcome are described. 
The last section concludes.

This paper is an extended version of the IJCAI'16 paper~\cite{GeorgeW16}.

\section{The Consistency Problem Based on Hierarchical Models}\label{sec: PCP}
For a set of preference statements, consistency is defined to be the property that the statements do not contradict each other. To formally define this term, we first describe the concept of hierarchical models. 

\subsection{Hierarchical Models}\label{sec:hierarchical models}

\noindent Hierarchical models will from here on be called HCLP models, where HCLP stands for ``Hierarchical Constraint Linear Program'' and points out the resemblance of the hierarchical order of the evaluation functions in HCLP models to the order of soft constraints in an HCLP~\cite{WilsonBorning1993}. In the following, we define HCLP structures, HCLP models, and their implied order relation that is a kind of lexicographic order. 

\begin{definition}[HCLP structures] An HCLP structure is a triple $\langle \calA, \oplus, \calC \rangle$. 
Here, $\calA$ is a finite set of alternatives 
and $\calC$ is a set of evaluation functions 
from the alternatives $\calA$ to the non-negative rational numbers $\mathds{Q}^{\geq 0}$.  
$\oplus$ is an associative, commutative, and strictly monotonic operation on $\mathds{Q}^{\geq 0}$, 
i.e., $x \oplus y < z \oplus y$ if and only if $x < z$. 
\end{definition}

In an HCLP structure, the evaluations $\calC$ as well as their combinations by $\oplus$ provide ratings of the alternatives due to unfavorable aspects, e.g., costs, for which we assume that $0$ is the best value. 
The notion of HCLP structures was first introduced in~\cite{WilsonGOSIJCAI2015longer} with $\oplus$ as an associative, commutative and monotonic operation. In this paper, we consider $\oplus$ to be a \emph{strictly} monotonic operation 
as this yields interesting properties allowing us to formulate fast algorithms for checking consistency.
Furthermore, we assume the operation to be computable in logarithmic time.
These assumptions still allow interesting operators like addition and multiplication which seem natural for combining aspects like costs and distances, but strict monotonicity excludes a minimum or maximum operator which could be desired sometimes.

For a subset $C \subseteq \calC$ of evaluations, we define the weak order (transitive and complete binary order) $\preccurlyeq_{C}^{\oplus}$ on the set of alternatives in the following way: for $\alpha, \beta \in \calA$, $\alpha \preccurlyeq_{C}^{\oplus} \beta$ if and only if $\bigoplus_{c \in C} c(\alpha) \leq \bigoplus_{c \in C} c(\beta)$. The corresponding strict order $\prec_{C}^{\oplus}$ is given by $\alpha \prec_{C}^{\oplus} \beta$ if and only if $\alpha \preccurlyeq_{C}^{\oplus} \beta$ and $\alpha \not \succcurlyeq_{C}^{\oplus} \beta$, i.e., $\bigoplus_{c \in C} c(\alpha) < \bigoplus_{c \in C} c(\beta)$. Then, the equivalence relation $\equiv_{C}^{\oplus}$ is given by $\alpha \equiv_{C}^{\oplus} \beta$ if and only if $\alpha \preccurlyeq_{C}^{\oplus} \beta$ and $\alpha \succcurlyeq_{C}^{\oplus} \beta$, i.e., $\bigoplus_{c \in C} c(\alpha) = \bigoplus_{c \in C} c(\beta)$. For $C = \emptyset$, $\alpha \equiv_{C}^{\oplus} \beta$ for all $\alpha, \beta \in \calA$.

\begin{definition}[HCLP models] An HCLP model $H$ for an HCLP structure $\langle \calA, \oplus, \calC \rangle$ is an ordered partition of a subset $\sigma(H) \subseteq \calC$ of evaluations. We write $H$ as sequence $(C_1, \dots, C_k)$, where the (possibly empty) sets $\emptyset \subseteq C_1, \dots, C_k  \subseteq \calC$ are disjoint and $\bigcup_{i =1, \dots, k} C_k~=~\sigma(H)$. We say $C_i$ is the $i$-th level set in $H$. We denote the empty HCLP model with $\sigma(H)~=~\emptyset$ by $H=()$.
\end{definition}

An HCLP model can be viewed as a hierarchy on the evaluation functions. For HCLP model $H = (C_1, \dots, C_k)$ the level set $C_1$ contains the most important evaluation functions; the level set $C_2$ contains the second most important evaluation functions and so on. Evaluations $\calC \setminus \sigma(H)$ that are not included in the HCLP model are irrelevant for rating the alternatives. Accordingly, we compare two alternatives first on a combination by $\oplus$ of the most important evaluation functions; only if these are equal do we compare them on the combination of the next most important evaluations.
Hence, each HCLP model $H = (C_1, \dots, C_k)$ implies a weak order $\preccurlyeq_{H}^{\oplus}$ on the alternatives that is a lexicographic order on combinations of evaluations. More specifically, for two alternatives $\alpha, \beta \in \calA$, $\alpha \preccurlyeq_{H}^{\oplus} \beta$ if and only if 
\begin{itemize}
	\item[(I)] for all $i = 1, \dots, k$, $\alpha \equiv_{C_i}^{\oplus} \beta$; or
	\item[(II)] there exists some $i \in \{1, \dots, k\}$ such that 
	$\alpha \prec_{C_i}^{\oplus} \beta$ and for all $1 \leq j < i$, $\alpha \equiv_{C_j}^{\oplus} \beta$.
\end{itemize}
Similarly to $\prec_{C}^{\oplus}$ and $\equiv_{C}^{\oplus}$, we define the strict weak order $\prec_{H}^{\oplus}$ and the equivalence relation $\equiv_{H}^{\oplus}$. For $\alpha, \beta \in \calA$, $\alpha \prec_{H}^{\oplus} \beta$ if and only if $\alpha \preccurlyeq_{H}^{\oplus} \beta$ and $\alpha \not \succcurlyeq_{H}^{\oplus} \beta$ (i.e., condition (II) holds). Analogue, $\alpha \equiv_{H}^{\oplus} \beta$ if and only if $\alpha \preccurlyeq_{H}^{\oplus} \beta$ and $\alpha \succcurlyeq_{H}^{\oplus} \beta$ (i.e., condition (I) holds).
Note, that by these definitions level sets of HCLP models can be empty, but empty level sets do not affect the relations $\preccurlyeq_{H}^{\oplus}$, $\prec_{H}^{\oplus}$ and $\equiv_{H}^{\oplus}$ or any statements based on these relations.

In this paper, we consider special classes of HCLP models where the sizes of the level sets are bounded by some constant. The class $\calC(t)$ is defined to be the set of HCLP models with level sets that contain at most $t$ evaluations, i.e., if $H=(C_1, \dots, C_k)$ is in $\calC(t)$ then $|C_i| \leq t$ for all $i=1, \dots,k$. Note, that the model classes satisfy the relation $\calC(s) \subseteq \calC(t)$ for $s \leq  t$. Class $\calC(1)$ contains standard lexicographic models.

\begin{ex}
Consider the choice of alternatives Apple Pie (AP), Chocolate Cake (CC) and Ice Cream (IC). The desserts are rated by the evaluation functions: calories ($c$), sugar ($s$), and fat ($f$). The values of $\calC=\{c,s,f\}$ are percentages of the recommended daily intake as shown in Table~\ref{tab: ex}. Here, 0 is the best possible value, meaning 0\% sugar, calories or fat of the recommended daily intake is contained in the dessert. Let $\oplus$ be the standard addition on $\mathds{Q}$.

\begin{table}[ht]
\tiny
\centering
\resizebox{.516\columnwidth}{!}{
\begin{tabular}{l|c|c|c}
 & AP & CC & IC\\
\hline
\hline
$c$ & 10 & 13 & 11 \\ \hline
$s$ & 23 & 23 & 16 \\ \hline
$f$ & 20 & 17 & 24 \\ \hline
$f \oplus s$ & 43 & 40 & 40 
\end{tabular}
}
\caption{Values of $c,s,f, f \oplus s$ evaluated on AP, CC, IC.}
\label{tab: ex}
\end{table}

The HCLP model $H=(\{f,s\},\{c\})$ consists of the pair of most important evaluations $f$ and $s$ followed by the singleton next most important evaluation $c$. Hence, $H$ is in $\calC(2)$ but not in $\calC(1)$. Also, $H$ implies the relations:\\
IC $\prec_{H}^{\oplus}$ CC, since $(f($IC$) \oplus s($IC$),c($IC$))= (40, 11) <_{\textit lex} (40,13) = (f($CC$) \oplus s($CC$), c($CC$))$. Similarly, IC $\prec_{H}^{\oplus}$~AP and  CC $\prec_{H}^{\oplus}$~AP.
\end{ex}

\subsection{Preference Consistency}\label{sec:consistency}

In the following, we first describe the language of preference statements considered in this paper and then define PCP. For this, we define $\calL^{\calA}_{\leq}$ to be the set of non-strict preference statements $\alpha \leq \beta$ (meaning $\alpha$ is preferred to $\beta$) on alternatives $\alpha, \beta \in \calA$. Similarly, $\calL^{\calA}_{<}$ is the set of strict preference statements $\alpha < \beta$ (meaning $\alpha$ is strictly preferred to $\beta$) for $\alpha, \beta \in \calA$. Let $\calL^{\calA} = \calL^{\calA}_{\leq} \cup \calL^{\calA}_{<}$. We write a preference statement $\vphi \in \calL^{\calA}_{\leq}$ as $\alpha_\vphi \leq \beta_\vphi$, and $\vphi \in \calL^{\calA}_{<}$ as $\alpha_\vphi < \beta_\vphi$ for $\alpha_\vphi, \beta_\vphi \in \calA$. We denote the non-strict version of preference statements $\Gamma \subseteq \calL^{\calA}$ by $\Gammadestrict$, i.e., $\Gammadestrict = \{\alpha_\vphi \leq \beta_\vphi \;|\; \vphi \in \Gamma\}$.

\begin{definition}[Satisfaction of Preference Statements] Let $H$ be an HCLP model for HCLP structure $\langle \calA, \oplus, \calC \rangle$. 
We say $H$ satisfies a preference statement $\vphi \in \calL^{\calA}$, 
denoted $H \vDash^{\oplus} \vphi$, if 
 $\alpha_\vphi \preccurlyeq_{H}^{\oplus} \beta_\vphi$ for $\vphi \in \calL^{\calA}_{\leq}$; or 
 $\alpha_\vphi \prec_{H}^{\oplus} \beta_\vphi$ for $\vphi \in \calL^{\calA}_{<}$.
Furthermore, $H$ satisfies a set of preference statements $\Gamma \subseteq \calL^{\calA}$, denoted $H \vDash^{\oplus} \Gamma$, if $H$ satisfies all preference statements in $\Gamma$, i.e., $H \vDash^{\oplus} \vphi$ for all $\vphi \in \Gamma$.
\end{definition}

\begin{definition}[Consistency] Let $\langle \calA, \oplus, \calC \rangle$ be an HCLP structure and $t \in \{1, \dots, |\calC|\}$ a constant. We say $\Gamma \subseteq \calL^{\calA}$ is $\calC(t)$-consistent, if there exists an HCLP model $H \in \calC(t)$ such that $H \vDash^{\oplus} \Gamma$.
\end{definition}

By this definition, the empty model $H=()$ always satisfies non-strict statements, but never satisfies strict statements. 
Thus, $\Gamma \subseteq \calL^{\calA}_{\leq}$ is always consistent. 
It is easy to show that $\Gamma$ is $\calC(t)$-consistent if $\Gamma$ is $\calC(s)$-consistent for some $s < t$.
Here, the class of HCLP models $\calC(1)$ consists of  lexicographic models that imply normal lexicographic order relations. 
Thus, if preference statements $\Gamma$ are consistent with respect to lexicographic models, then $\Gamma$ is consistent with respect to HCLP models in $\calC(t)$. 
	Also, every induced order relation from an HCLP model can be represented by a weighted sum of evaluation functions that uses normalized weights with extreme values, see~\cite{GeorgeRW15}. 
	The corresponding preference models for order relations based on weighted sums are called weighted average models. 	
	George et al. \cite{GeorgeRW15} show that,  if preference statements $\Gamma$ are consistent with respect to HCLP models in $\calC(t)$, 
	then $\Gamma$ is consistent with respect to weighted average models. 

\begin{ex}[continued]
Consider preference statements:
\begin{enumerate}
	\item[(1)] I strictly prefer ice cream to apple pie (IC~$<$~AP).
	\item[(2)] I prefer chocolate cake to apple pie (CC $\leq$ AP).
\end{enumerate}
Then (1) is a strict preference statement, (2) is a non-strict preference statement. As before, HCLP model $H=(\{c,s\},\{f\})$ satisfies IC $\prec_{H}^{\oplus}$ AP and CC $\prec_{H}^{\oplus}$ AP and so CC~$\preceq_{H}^{\oplus}$ AP. Thus, (1) and (2) together are consistent.
\end{ex}

We can now formulate the Preference Consistency and Deduction (decision) Problems for classes $\calC(t)$.

\smallskip\noindent\textbf{$\calC(t)$ {Preference Consistency Problem} ($\calC(t)$-PCP):} Given an HCLP structure $\langle \calA, \oplus, \calC \rangle$, a constant $t \in \{1, \dots, |\calC|\}$ and a set of preference statements $\Gamma \subseteq \calL^{\calA}$. Is $\Gamma$ $\calC(t)$-consistent?

\smallskip\noindent\textbf{$\calC(t)$ {Preference Deduction Problem} ($\calC(t)$-PDP):} Given an HCLP structure $\langle \calA, \oplus, \calC \rangle$, a constant $t \in\{1, \dots, |\calC|\}$, some preference statements $\Gamma \subseteq \calL^{\calA}$ and $\vphi \in \calL^{\calA}\setminus\Gamma$. Does $H \vDash^{\oplus}~\vphi$ hold for all $H \in \calC(t)$ with $H \vDash^{\oplus} \Gamma$?

Wilson et al.~\cite{WilsonGOSIJCAI2015longer} (in their Proposition 1) show that $\calC(t)$-PCP and $\calC(t)$-PDP are mutually expressive, i.e., PCP can be solved by solving PDP and vice versa:

\begin{prop}\label{pr:basic-cons-deduction}
$H \vDash^{\oplus} \vphi$ for all $H \in \calC(t)$ with $H \vDash^{\oplus} \Gamma$ if and only if
$\Gamma\cup \{\neg\vphi\}$ is $\calC(t)$-inconsistent, where $\neg\vphi = \beta_\vphi <~\alpha_\vphi$ for $\vphi \in \calL^{\calA}_{\leq}$, or $\neg\vphi = \beta_\vphi \leq \alpha_\vphi$ for $\vphi \in \calL^{\calA}_{<}$.
\end{prop}

They also established that $\calC(t)$-PDP is coNP-complete for any $t \geq 2$ (even if $\oplus$ is strictly monotonic). Thus $\calC(t)$-PCP is NP-complete. A greedy algorithm can solve the special cases $\calC(1)$-PCP and $\calC(1)$-PDP in time $O(|\calC|\cdot|\Gamma|)$.
	Starting with the empty model, the HCLP model is constructed by repeatedly adding evaluations as singleton level sets 
	such that the preference statements $\Gamma$ are not opposed. 
	The resulting HCLP model satisfies all strict preference statements in $\Gamma$ 
	if and only if $\Gamma$ is $\calC(1)$-consistent. 
	The correctness of this algorithm strongly depends on the fact 
	that all maximal $\Gammadestrict$-satisfying $\calC(1)$  HCLP model contain the same evaluations and (strictly) satisfy the same statements in $\Gamma$. 
	This only holds for the class $\calC(1)$, and not for the general case of $\calC(t)$. 
	In the remainder of this paper, we concentrate on the NP-complete case, i.e., $t \geq 2$ .

\section{MILP Formulation}\label{sec: MILP}
We describe a Mixed Integer Linear Programming (MILP) formulation for $\calC(t)$-PCP with HCLP structure $\langle \calA, \oplus, \calC \rangle$ and preference statements $\Gamma \subseteq \calL^{\calA}$, where evaluations $\calC$ are integral, $\oplus$ is the standard addition on integers and $n = |\calC|$.

\smallskip\noindent\textbf{Assigning Evaluations to Level Sets:}
We introduce a matrix of Boolean variables $Y \in \{0,1\}^{n \times n}$ such that $y_{i,j} = 1$ if and only if evaluation $i$ is included in the $j$-th level set of the solution HCLP model. For $\calC(t)$-PCP, every evaluation is contained in at most one level set and the cardinality of the level sets is bounded by $t$.
\begin{equation}
	\sum_{j = 1}^{n} y_{i,j} \leq 1\;\; \text{ and }
	\sum_{j = 1}^{n} y_{j,i} \leq t \;\; \forall i = 1, \dots, n.
\end{equation}

\smallskip\noindent\textbf{Maintaining Values of $\oplus$-combined Level Sets:}
The matrix of integer variables $X \in \mathds{Q}^{n \times |\Gamma|}$ contains the values of the combined evaluation functions in the level sets for the alternatives of the preference statements. Thus, $x_{i, \vphi} = \bigoplus_{c \in C_i} c(\alpha_\vphi) - \bigoplus_{c \in C_i} c(\beta_\vphi)$ maintains the degree of support/opposition of statement $\vphi$ in the $i$-th level set $C_i$.
\begin{equation}
	\sum_{i = 1}^{n} y_{i,j}( c(\alpha_\vphi) - c(\beta_\vphi)) = x_{j, \vphi} \;\; \forall j = 1, \dots, n, \forall \vphi \in \Gamma.
\end{equation}

We define the bounds $M_\vphi \geq x_{j, \vphi} \geq m_\vphi$ for all $x_{j, \vphi}$ by
\begin{equation}
m_\vphi = \min \limits_{E \subseteq \calC} {\sum \limits_{c \in E}  c(\alpha_\vphi) - c(\beta_\vphi)} 
= \hspace{-.7cm} \sum \limits_{c \in \calC, c(\alpha_\vphi) < c(\beta_\vphi)} \hspace{-.5cm} c(\alpha_\vphi) - c(\beta_\vphi), \nonumber 
\end{equation}
\begin{equation}
M_\vphi = \max \limits_{E \subseteq \calC} {\sum \limits_{c \in E}  c(\alpha_\vphi) - c(\beta_\vphi)} 
= \hspace{-.7cm} \sum \limits_{c \in \calC, c(\alpha_\vphi) > c(\beta_\vphi)} \hspace{-.5cm} c(\alpha_\vphi) - c(\beta_\vphi). \nonumber
\end{equation}

\smallskip\noindent\textbf{Maintaining the Sign of Level Sets (Supporting, Opposing and Indifferent):}
The Boolean variables $s^{<0}_{j,\vphi}, s^{>0}_{j,\vphi}$ and $s^{=0}_{j,\vphi}$ express the sign for $x_{j, \vphi}$. 
 This is, $s^{<0}_{j,\vphi} = 1$ if and only if $x_{j, \vphi} <0$, $s^{>0}_{j,\vphi} = 1$ if and only if $x_{j, \vphi} >0$, and $s^{=0}_{j,\vphi} = 1$ if and only if $x_{j, \vphi} =0$.
Since exactly one of the relations holds, 
\begin{equation}
	s^{<0}_{j,\vphi} + s^{>0}_{j,\vphi} + s^{=0}_{j,\vphi}= 1 \;\; \forall j = 1, \dots, n, \forall \vphi \in \Gamma.
\end{equation}
 
 To enforce the equivalences, we make use of the bounds $M_\vphi$ and $m_\vphi$ and the integrity of the evaluations. In particular, we use that the lowest positive value $x_{j, \vphi}$ can take is $1$ and the highest negative value is $-1$.
 
For the implication $s^{<0}_{j,\vphi} = 1$ $\Rightarrow$ $x_{j, \vphi}<0$, we formulate:
\begin{equation}
	x_{j, \vphi} + s^{<0}_{j,\vphi}(M_\vphi+1) \leq M_\vphi \;\; \forall j = 1, \dots, n, \forall \vphi \in \Gamma.
\end{equation}

For the implication $s^{>0}_{j,\vphi} = 1$ $\Rightarrow$ $x_{j, \vphi}>0$, we formulate:
 \begin{equation}
	x_{j, \vphi} + s^{>0}_{j,\vphi}(m_\vphi-1) \geq m_\vphi \;\; \forall j = 1, \dots, n, \forall \vphi \in \Gamma.
\end{equation}

Finally, we enforce $s^{=0}_{j,\vphi} = 1$ $\Rightarrow$ $x_{j, \vphi}=0$ by
 \begin{equation}
	x_{j, \vphi} - (1-s^{=0}_{j,\vphi})m_\vphi \geq 0 \;\; \forall j = 1, \dots, n, \forall \vphi \in \Gamma \text{ and }
\end{equation}
 \begin{equation}
	x_{j, \vphi} - (1-s^{=0}_{j,\vphi})M_\vphi \leq 0 \;\; \forall j = 1, \dots, n, \forall \vphi \in \Gamma.\;\;\;\;\;\;\;
\end{equation}

The equivalences follow from (3) together with (4)-(7).

\smallskip\noindent\textbf{Satisfaction of Strict and Non-strict Statements:}
 Following the definition of $\preccurlyeq_{H}^{\oplus}$, the HCLP model corresponding to the variable assignments of $Y$ satisfies a non-strict statement $\vphi$ in $\Gamma$ if and only if
\begin{itemize}
	\item[(I$'$)] for all $i = 1, \dots, n$, $s^{=0}_{i,\vphi} = 1$; or
	\item[(II$'$)] there exists some $i \in \{1, \dots, n\}$ such that 
	$s^{<0}_{i,\vphi} = 1$ and for all $1 \leq j < i$, $s^{=0}_{j,\vphi} = 1$.
\end{itemize}
Analogously, a strict statements $\vphi$ in $\Gamma$ is satisfied if and only if (II$'$) holds.
It is easy to check that conditions (I$'$) or (II$'$) holding for all $\vphi \in \Gamma$ is equivalent to 
 \begin{equation}
	 \sum \limits_{j=1}^{i-1} s^{<0}_{j,\vphi} \geq s^{>0}_{i,\vphi}\;\; \forall i = 1, \dots, n, \forall \vphi \in \Gamma.
\end{equation}
Inequality (8) yields the satisfaction of $\Gammadestrict$.
We enforce satisfaction of all strict statements in $\Gamma$, by
 \begin{equation}
	 \sum \limits_{j=1}^{n} s^{<0}_{j,\vphi} \geq 1 \;\; \forall \vphi \in \Gamma \cap \Lstrict.
\end{equation}

The constraints (1)-(9) form a rather simple MILP formulation for $\calC(t)$-PCP. 
Constraints (3)-(9) could be replaced by sums with extreme weights to enforce a lexicographic order on the level sets.	
	Let $L>0$ be sufficiently large, then the following inequalities can be used to replace (3)-(9).
	 \begin{equation}
		 \sum \limits_{j=1}^{n} \frac{x_{j,\vphi}}{L^j} \leq 0 \;\; \forall \vphi \in \Gamma \cap \Lnonstrict.
	\end{equation}
	 \begin{equation}
		 \sum \limits_{j=1}^{n} \frac{x_{j,\vphi}}{L^j} < 0 \;\; \forall \vphi \in \Gamma \cap \Lstrict.
	\end{equation}
 However, these inequalities lead to numerical difficulties for the solver. 
  	This is true even for small instances with integral evaluation functions of small domains and a sophisticated choice for $L$. 
Also, decision variables $y_{i,j}$ could be substituted by $y'_{i,j}$ such that $y'_{i,j}=1$ if and only if $i$ is included in a level set $\geq j$. Variables $s^{<0}_{j,\vphi}, s^{>0}_{j,\vphi},s^{=0}_{j,\vphi} \in \{0,1\}$ might be replaceable by a variable $s_{j,\vphi} \in \{0,1,2\}$.
However, since our MILP is a satisfaction problem, not an optimization problem, it is not clear whether any of these measures improve the formulation.
After trying various Constraint Programming models with set or binary variables, different versions of constraints and different search heuristics, the MILP formulation using inequalities (1)-(9) seemed most promising.

\section{Recursive Algorithms}\label{sec: rec algo}
In the following, we describe two recursive search algorithms for $\calC(t)$-PCP. 
The algorithms are based on properties of PCP that can be used to prune the search space. Both try to construct a $\Gamma$-satisfying HCLP model by sequentially adding new level sets to the model that do not oppose any preference statement that is not strictly satisfied so far. 
Thus, during the algorithm the current model always satisfies $\Gammadestrict$, the non-strict version of $\Gamma$. 
We backtrack when the current model cannot be extended further and the model does not satisfy all strict preference statements.
The approaches aim to reduce the number of  $\Gammadestrict$-satisfying HCLP models constructed by the algorithm. In particular, they try to identify and ignore level sets which cannot lead to a $\Gamma$-satisfying HCLP model although not opposing the preference statements.

\smallskip\noindent\textbf{Utilising Sequences of Singleton Level Sets:}
The first approach is based on the idea of including as many singleton level sets as possible. This seems computationally less challenging since a $\Gammadestrict$- satisfying sequence of singleton level sets that is maximal in the number of level sets can be found in time $O(|\calC|\cdot|\Gamma|)$~\cite{WilsonGOSIJCAI2015longer}. Remember, that $\Gammadestrict$ is defined to be the non-strict version of $\Gamma$, i.e., $\Gammadestrict = \{\alpha_\vphi \leq \beta_\vphi \;|\; \vphi \in \Gamma\}$.
In the following we show that for strictly monotonic operators $\oplus$ the recursive search algorithm never needs to backtrack over the choice of such singleton sequences.
We first establish the following property for strictly monotonic operators $\oplus$ which can be shown by a short technical proof.

\begin{lem}\label{lem: strict mon rule}
	Let $\oplus$ be a strictly monotonic operator and let $X,Y \subseteq \calC$ be sets of evaluation functions with $X \subseteq Y$. Let $\alpha, \beta \in \calA$ be alternatives such that $X$ is indifferent under $\alpha$ and $\beta$, i.e., $\alpha \equiv_X^\oplus \beta$. Then $\alpha \prec_{Y}^\oplus \beta$ if and only if $\alpha \prec_{Y \setminus X}^\oplus \beta$. Hence, $\alpha \equiv_{Y}^\oplus \beta$ if and only if $\alpha \equiv_{Y \setminus X}^\oplus \beta$.
\end{lem}
	\begin{proof}
		Let $\alpha \prec_{Y}^\oplus \beta$, i.e., $\bigoplus_{c \in Y} c(\alpha) < \bigoplus_{c \in Y} c(\beta)$. 
		Since $\oplus$ is associative and commutative, and $X \subseteq Y$, 
		this is equivalent to $\bigoplus_{c \in Y\setminus X} c(\alpha) \oplus \bigoplus_{c \in X} c(\alpha)< \bigoplus_{c \in Y\setminus X} c(\beta) \oplus \bigoplus_{c \in X} c(\beta)$. 
		By strict monotonicity of $\oplus$ and because $X$ is indifferent under $\alpha$ and $\beta$, 
		this is equivalent to $\bigoplus_{c \in Y\setminus X} c(\alpha) < \bigoplus_{c \in Y\setminus X} c(\beta)$, i.e., $\alpha \prec_{Y \setminus X}^\oplus \beta$.
		Analogously, we can prove $\alpha \succ_{Y}^\oplus \beta$ if and only if $\alpha \succ_{Y \setminus X}^\oplus \beta$.
		Both equivalences together yield $\alpha \equiv_{Y}^\oplus \beta$ if and only if $\alpha \equiv_{Y \setminus X}^\oplus \beta$.
	\end{proof}
	Note, that the previous proof explicitly uses the strict monotonicity of the operator $\oplus$. 
	
We define the (non-commutative) combination of two HCLP models $H=(C_1, \dots, C_l)$ and $H'=(C'_1, \dots, C'_k)$ in $C(t)$ as $H \circ H' = (C_1, \dots, C_l, (C'_1 \setminus \sigma_H), \dots, (C'_k \setminus \sigma_H))$, where $\sigma_H = \bigcup_{i=1, \dots, l} C_i$.
It is easy to see that $H \circ H'$ is an HCLP model in $C(t)$ as well.
The following proposition shows how the satisfaction of preference statements $\Gamma$ from $H'$ persists under combination with sequences of singleton level sets that only satisfy $\Gammadestrict$.

\begin{prop}\label{prop: singletons-first}
	Let $(\langle \calA, \oplus, \calC \rangle, \Gamma)$ be an instance of $C(t)$-PCP. If $H=(c_1, \dots, c_l)$ is a $\Gammadestrict$- satisfying model in $\calC(1)$ and $H' = (C'_1, \dots, C'_k)$ is a $\Gamma$-satisfying model in $C(t)$, then $H \circ H'$ is a $\Gamma$-satisfying model in $C(t)$.
\end{prop}
\begin{proof}
	We show, that $H \circ H'$ satisfies $\Gammadestrict$ and strictly satisfies the preference statements that $H'$ strictly satisfies. 
	Hence, $H \circ H'$ is a $\Gamma$-satisfying HCLP model in $C(t)$.
	
	Recall that a preference statement $\vphi$ is strictly satisfied when there exists a level set $C$ supporting $\vphi$, i.e.,  $\alpha_\vphi \prec_C \beta_\vphi$, 
	and all preceding level sets $C'$ are indifferent under $\vphi$, i.e., $\alpha_\vphi \equiv_{C'} \beta_\vphi$. 
	Hence, the preference statements in $\Gamma$ that are strictly satisfied by $H = (c_1, \dots, c_l)$ are also strictly satisfied by $H \circ H' = (c_1, \dots, c_l, (C'_1 \setminus \sigma_H), \dots, (C'_k \setminus \sigma_H))$. 
	Let $\Gamma'$ be the set of remaining preference statements that are not strictly satisfied by $H$. 
	Since $H$  satisfies $\Gammadestrict$,
	$H$ is indifferent under all statements in $\vphi \in \Gamma'$, 
	i.e., $c_i(\alpha_\vphi) = c_i(\beta_\vphi)$ for all $1 \leq i \leq l$. 
	Consider an arbitrary level set $C$ in $H'$ and a preference statement $\vphi \in \Gamma'$. 
	Repeatedly applying Lemma~\ref{lem: strict mon rule} for the singleton level sets in $\sigma_H \cap C$ yields:
	$\alpha_\vphi \sim_{C}^\oplus \beta_\vphi$ if and only if $\alpha_\vphi \sim_{C \setminus \sigma_H}^\oplus \beta_\vphi$, where $\sim$ is one of the relations $\prec$, $\equiv$ or $\succ$. 
	Thus, the level sets $C'_i \setminus \sigma_H$ in $H \circ H'$ have the same relation towards statements $\vphi \in \Gamma'$ as the level sets $C'_i$ in $H'$. 
	Since the initial singleton sequence in $H \circ H'$ is indifferent under preference statements $\vphi \in \Gamma'$, $H \circ H'$ satisfies $\vphi$ if and only if $H'$ satisfies $\vphi$. 
	Also, all statements $\Gamma \setminus \Gamma'$ are strictly satisfied by $H \circ H'$. 
	Hence, $H \circ H'$ satisfies $\Gammadestrict$ and strictly satisfies all statements in $\Gamma$ that $H'$ strictly satisfies. 
\end{proof}

Proposition~\ref{prop: singletons-first} immediately leads to the next result.

\begin{prop}\label{cor: pref max singleton sequences}
	Let $H$ be a $\Gammadestrict$-satisfying model in $C(1)$ that consists of a maximal number of singleton level sets. If $\Gamma$ is $C(t)$-consistent, then there exists a $\Gamma$-satisfying model in $C(t)$ with $H$ as initial sequence.
\end{prop}

Based on Proposition~\ref{cor: pref max singleton sequences}, we describe the algorithm PC-check that solves $C(t)$-PCP by trying to construct a $\Gamma$-satisfying HCLP model. This method is summarised in the algorithm below disregarding the framed parts.
After finding an initial singleton sequence $(c_1, \dots, c_k)$ that is maximal while satisfying $\Gammadestrict$ (in time $O(|\calC|\cdot|\Gamma|)$ by a greedy algorithm~\cite{WilsonGOSIJCAI2015longer}), 
we consider every possible (not opposing) level set $C$ of size $2 \leq |C| \leq t$. 
Let $\Gamma'$ be the set of preference statements in $\Gamma$ that are not strictly satisfied by $H=(c_1, \dots, c_k, C)$. 
We try to extend the sequence $H$ by another $\Gamma'$-satisfying HCLP model. 
We construct this extending model by recursively calling the method for the subproblem with statements $\Gamma'$ and evaluations $\calC'=\calC - \{c_1, \dots, c_k\} - C$. 
If no such extension exists (that satisfies $\Gamma$), we backtrack over the last chosen level set $C$ and try a new level set. 
Note, that by Proposition~\ref{cor: pref max singleton sequences}, we never have to backtrack over the choice of singleton level sets, which can be a significant advantage over solving the MILP model.
As soon as the currently considered sequence in the algorithm satisfies $\Gamma$ we stop and return the sequence, i.e., showing that the instance is $C(t)$-consistent.
If no $\Gamma$-satisfying sequence can be found after exploiting all possible ($\Gammadestrict$-satisfying) HCLP model, we stop and return that the instance is $C(t)$-inconsistent.

\smallskip\noindent\textbf{Maintaining Conflicting Sets:}
In the following, we extend the algorithm PC-check($\calC, \Gamma, \oplus, t$) by maintaining conflicting sets that cannot be contained in the later level sets, and thus reduce the number of backtracks. 
Proposition~\ref{prop: conf-set-property} shows that the satisfaction of $\Gamma$ persists for an HCLP model $H'$ when combining with an HCLP model $H$ that extends an initial sequence of level sets of $H'$ by one more level set and only satisfies $\Gammadestrict$.

\begin{prop}\label{prop: conf-set-property}
	Let $(\langle \calA, \oplus, \calC \rangle, \Gamma)$ be an instance of $C(t)$-PCP that is $C(t)$-consistent. 
	Let $H = (C_1, \dots, C_k, B)$ be a $\Gammadestrict$-satisfying HCLP model in $C(t)$ and 
	let $H' = (C_1, \dots, C_k,C_{k+1}, \dots, C_l)$ be a $\Gamma$-satisfying HCLP model in $C(t)$ with $B \subseteq C_j$ for some $k+1 \leq j \leq l$.
	Then $H \circ H'  = (C_1, \dots, C_k, B, C_{k+1}, \dots, (C_j \setminus B), \dots, C_l)$ is a $\Gamma$-satisfying HCLP model in $C(t)$.
\end{prop}
	\begin{proof}
		We show that the HCLP model $H \circ H' =(C_1, \dots, C_k,B,C_{k+1}, \dots,C_{j-1}, C_j \setminus B, C_{j+1}, \dots, C_l)$ satisfies $\Gammadestrict$ 
		and strictly satisfies all preference statements that are strictly satisfied by $H'$.
		Hence, $H \circ H'$ is a $\Gamma$-satisfying HCLP model in $C(t)$.
		
		A preference statement $\vphi$ is strictly satisfied when there exists a level set $C$ supporting $\vphi$, i.e.,  $\alpha_\vphi \prec_C \beta_\vphi$, 
		and all preceding level sets $C'$ are indifferent under $\vphi$, i.e., $\alpha_\vphi \equiv_{C'} \beta_\vphi$. 
		Hence, the preference statements in $\Gamma$ that are strictly satisfied by $H = (C_1, \dots, C_k, B)$ are also strictly satisfied by $H \circ H'$. 
		In the following, we consider all remaining preference statements. 
		Let $\Gamma'$ be the set of preference statements that are not strictly satisfied by $H$. 
		Since $H$ satisfies $\Gammadestrict$,
		the sequence $(C_1, \dots, C_k, B)$ is indifferent under all statements in $\Gamma'$, 
		i.e., $\bigoplus_{c \in C} c(\alpha_\vphi) = \bigoplus_{c \in C} c(\beta_\vphi)$ for all $\vphi \in \Gamma'$ and $C \in \{C_1, \dots, C_k, B\}$. 
		The level sets $C_i$ with $k+1 \leq i \leq l$ and $i \neq j$ are level sets in both $H'$ and $H \circ H'$. 
		Hence, for the satisfaction of statements in $\Gamma'$ we only need to compare the level set $C_j$ in $H'$ to the level set $C_j \setminus B$ in $H \circ H'$.
		Consider a preference statement $\vphi \in \Gamma'$.
		Since $B$ is indifferent under $\vphi$, by Lemma~\ref{lem: strict mon rule}, 
		$\alpha_\vphi \sim_{C_j}^\oplus~\beta_\vphi$ if and only if $\alpha_\vphi \sim_{C_j \setminus B}^\oplus \beta_\vphi$, 
		where $\sim$ is one of the relations $\prec$, $\equiv$ or $\succ$. 
		Thus, for $\vphi \in \Gamma'$ all level sets $C_{k+1}, \dots, C_l$ in $H$ have the same relation towards $\vphi$ 
		as the level sets $C_{k+1}, \dots,C_{j-1}, C_j \setminus B, C_{j+1}, \dots, C_l$ in $H'$. 
		Since the initial sequence $(C_1, \dots, C_k, B)$ is indifferent under preference statements $\vphi \in \Gamma'$, $H \circ H'$ satisfies $\vphi$ if and only if $H'$ satisfies $\vphi$. 
		Furthermore, all statements $\Gamma \setminus \Gamma'$ are strictly satisfied by $H \circ H'$. 
		We have shown, that $H \circ H'$ satisfies $\Gammadestrict$ and strictly satisfies all preference statements that $H'$ strictly satisfies. 
	\end{proof}
	
	Reformulating Proposition~\ref{prop: conf-set-property} yields the following statement.
	\begin{prop}\label{cor: conf sets}
		Let $H = (C_1, \dots, C_k, B)$ be a $\Gammadestrict$-satisfying HCLP model in $C(t)$. 
		If there exists no extension $(C_1, \dots, C_k, B, C_{k+1}, \dots, C_{l})$ of $H$ in $C(t)$ 
		that satisfies $\Gamma$, then for all $\Gamma$-satisfying HCLP model $H' = (C_1, \dots, C_k,C_{k+1}, \dots, C_l)$  in $C(t)$, 
		$B \nsubseteq C_j$ for all $k+1 \leq j \leq l$.
	\end{prop}
	
	This proposition characterises the conflicting sets $B$ that are maintained in the second recursive approach. 
	By Proposition~\ref{cor: conf sets}, no $\Gamma$-satisfying HCLP model that extends $(C_1, \dots, C_k)$ can contain the conflicting set $B$.
	Thus, at a point of the algorithm where we backtrack because no $\Gamma$-satisfying extension of the current HCLP model can be found
	we add the last considered level set to the list of conflicting sets and then choose a new next level set that does not contain any conflicting set.
	This extension of the algorithm PC-check($\calC, \Gamma, \oplus, t$) is given as the framed parts in the algorithm below.
	Here, although reducing the search space, we have to maintain a list of conflicting sets which can grow exponentially large. 
	Thus, it is not obvious if maintaining conflicting sets is advantageous. 
	We introduce the additional parameter $s$ which is a cardinality bound on the size of the conflicting sets. 
	Only conflicting sets $C$ with $|C| \leq s$ are maintained, such that the space needed is $O(s \cdot \binom{n}{s})$ for $s \leq n/2$, and $O(s \cdot \binom{n}{n/2})$ otherwise.

\begin{figure}[ht]
\noindent\fbox{
\begin{minipage}{.945\columnwidth}
 \textbf{Algorithm:} PC-check($\calC, \Gamma, \oplus, t,$ \mybox{$\calS=\emptyset,  s$} )\\
$H \leftarrow (c_1, \dots, c_k)$ some $\Gammadestrict$-satisfying singleton sequence with $k$ maximal\;\\
\textbf{if} $H \vDash \Gamma$ \textbf{then}  \textbf{return} {$H$}\;\\
\textbf{for all} $C \subseteq \calC - \{c_1, \dots, c_k\}$ with $2~\leq~|C|~\leq~t$ and $\alpha_\vphi \preccurlyeq_{C}^{\oplus} \beta_\vphi$ for all $\vphi$ with $\alpha_\vphi \equiv_H^\oplus \beta_\vphi$
\mybox{such that there exists no $S \in \calS$ with $S \subseteq C$\;} \textbf{do}

\hspace{.3cm}$H \leftarrow (c_1, \dots, c_k,C)$\;

\hspace{.3cm}\textbf{if} $H \vDash \Gamma$ \textbf{then}  \textbf{return} {$H$}

\hspace{.3cm}\textbf{else} \textbf{if} $H \vDash \Gammadestrict$ \textbf{then}

\hspace{.53cm}$\Gamma'  = \{\vphi \in \Gamma \;|\; \alpha_\vphi \equiv_H^\oplus \beta_\vphi\}$, $\calC' = \calC - \sigma(H)$\;

\hspace{.53cm}$H \leftarrow (c_1, \dots, c_k,C,$PC-check($\calC', \Gamma', \oplus, t,$ \mybox{$\calS,  s$} ))\;

\hspace{.53cm}\textbf{if} $H \vDash \Gamma$ \textbf{then} \textbf{return} {$H$} \mybox{\textbf{else} $\calS \leftarrow \calS \cup {C}$ for $|C|\leq s$\;}\\
 \textbf{return} \emph{Inconsistent}
 \end{minipage}
}
\end{figure}

\begin{ex}
Consider a $\calC(3)$-PCP instance with the following evaluation functions $c_1, \dots, c_5$ and statements on $\alpha, \beta, \gamma, \delta$. 
 \begin{table}[h!]
\centering
\resizebox{.55\columnwidth}{!}{
\begin{tabular}{l|c|c|c|c|c|c|c}
 & $\alpha$ & $\leq$ & $\beta$ & $\leq$ & $\gamma$ & $<$ & $\delta$ \\
\hline
\hline
$c_1$ & 1 & & 0 & & 0 & & 0 \\ \hline
$c_2$ & 0 & & 2 & & 2 & & 2 \\ \hline
$c_3$ & 1 & & 1 & & 0 & & 1 \\ \hline
$c_4$ & 0 & & 2 & & 1 & & 1 \\ \hline
$c_5$ & 2 & & 0 & & 1 & & 0 
\end{tabular}
}
\label{tab: ex2}
\end{table}

\noindent Let $\oplus$ be the standard addition on $\mathds{Z}$. Suppose, in the first step PC-check finds the maximal singleton sequence $(c_2, c_1)$ (which cannot be extended by any other evaluation without violating $\Gammadestrict$). Then the algorithm will in turn consider sets $\{c_3,c_4\}$, $\{c_3,c_5\}$, $\{c_4,c_5\}$ and  $\{c_3,c_4,c_5\}$. The sequences $c_2,c_1,\{c_3,c_4\}$ and $c_2,c_1,\{c_4,c_5\}$ violate $\Gammadestrict$. The sequence $c_2,c_1,\{c_3,c_5\}$ satisfies $\Gammadestrict$ but cannot be extended to satisfy $\Gamma$. In PC-check($\calC, \Gamma, \oplus, t,\calS,  s$), the set $\{c_3,c_5\}$ is added to the conflicting sets $\calS$ and thus the set $\{c_3,c_4,c_5\}$ does not have to be checked (by Proposition~\ref{prop: conf-set-property}). PC-check($\calC, \Gamma, \oplus, t$) finds that $c_2,c_1,\{c_3,c_4,c_5\}$ violates $\Gammadestrict$. Thus none of the possible sets leads to a $\Gamma$-satisfying sequence and ``Inconsistent'' is returned. Note, that PC-check does not have to backtrack over the choice of evaluations in the initial singleton set sequence $c_2,c_1$ (by Proposition~\ref{cor: pref max singleton sequences}).
\end{ex}

\section{Experimental Runtime Comparisons}\label{sec: experiments}

In our experiments, we compare the approaches from Section~\ref{sec: MILP} and~\ref{sec: rec algo} for solving PCP by their running time. Here, the MILP formulation functions as a baseline and is expected to be outperformed by the two recursive approaches as they directly exploit the problem structure to perform less backtracks in a way that is not recognized by CPLEX. Note, that it is not obvious how to incorporate the pruning of the search space, that is preformed by the recursive algorithms, in a MILP model in form of constraints or heuristics. 
We investigate the degree of improvement of the recursive algorithms towards the rather simple MILP formulation and the relation of the recursive algorithms towards each other. Though PC-check($\calC, \Gamma, \oplus, t,\calS,  s$) prunes the search space further than PC-check($\calC, \Gamma, \oplus, t$), the list of maintained conflicting sets can grow exponentially large. Thus, it is not obvious if maintaining conflicting sets is advantageous.

\smallskip\noindent\textbf{Instances:}
For our experiments we considered different instance sizes, to observe the effect on the running time by varying the number of evaluations $n$ and the number of preference statements $g$. For the lack of real world data, we generated 50 instances uniformly at random with evaluation functions with domains $\{0,1,2,3,4,5\}$ for each of the problem sizes $n,g \in \{10, 15, \dots, 50\}$ where we fix the number of alternatives that the preference statements are based on to $m=25$. 
	First an $n \times m$ matrix is generated that gives the values of the evaluation functions for the alternatives. We next draw $g$ tuple of alternatives $(\alpha_i,\alpha_j)$ with $i<j$ uniformly at random (without repetition) such that the corresponding preference statement $\alpha_i \leq \alpha_j$ or $\alpha_i < \alpha_j$ coincides with the linear order $\alpha_1 < \dots < \alpha_m$. This way, we avoid cycles in the statements, which trivially lead to inconsistency. The first $\lceil g/2 \rceil$ statements are the strict statements, the remaining statements are the non-strict statements.	
Note, that not all alternatives generated are involved in preference statements. Thus, $m$ has no direct influence on the size of the search space or the running time. For time reasons, some experiments were not conducted for all instance sizes.

\smallskip\noindent\textbf{Implementation:}
We implemented all three approaches in Java Version 1.8 using the IBM ILOG CPLEX (version 12.6.2) library for the MILP formulation.
All experiments were conducted independently on a 2.66Ghz quad-core processor with 12GB memory.

 We choose $\oplus$ as the standard addition on the integers.
To reduce the number of experiments, we allow the cardinality bound on the level sets to be $t=n$, the number of evaluations, and fix the cardinality bound on the maintained conflicting sets to $s=5$ (which gives the bound $|\calS| \leq \binom{n}{s} \leq \binom{50}{5} = 2118760$). Since $\calC(k') \subseteq \calC(k)$ for all $k' < k$, we expect the running times to be lower for smaller $t$. Also, $\calC(k)=\calC(n)$ for all $k \geq n$, i.e., the running times are the same for bigger~$t$.
  	For the recursive algorithms we enumerate next level sets with lower cardinality before ones with higher cardinality, and level sets containing evaluations with smaller indexes before ones with higher indexes.

\smallskip\noindent\textbf{Experimental Results:}
As expected, solving the MILP formulation of PCP (as presented in Section~\ref{sec: MILP}) by the CPLEX solver is much slower than by the two recursive algorithms PC-check (as presented in Section~\ref{sec: rec algo}), see Table~\ref{tab: RT gap}.
However, it is remarkable how quickly the ratio between the mean times of solving the MILP and PC-check grows with the number of statements and evaluations in the instances.
\begin{table}[ht]
\centering
\resizebox{\columnwidth}{!}{
\begin{tabular}{l|l|c|c|c}
$g=$				& 								& $n=10$	& $n=15$	& $n=20$	\\ \hline \hline
\multirow{4}{*}{10}	& PC-check($\calC, \Gamma, \oplus, t$)			&0.011	&0.01		&0.04		\\ \cline{3-5}
				& PC-check($\calC, \Gamma, \oplus, t, \calS, s$)	&0.003	&0.01		&0.03		\\ \cline{3-5}
				& MILP							&0.22		&10.53	&149.81	\\ \cline{3-5}
				& minimal ratio: MILP/PC-check				&20		&1053		&3745.25	\\ \hline
\multirow{4}{*}{15}	& PC-check($\calC, \Gamma, \oplus, t$)			&0.003	&0.01		&0.29		\\ \cline{3-5}
				& PC-check($\calC, \Gamma, \oplus, t, \calS, s$)	&0.001	&0.01		&0.28		\\ \cline{3-5}
				& MILP							&0.22		&220.28	&$>$41472		\\ \cline{3-5}
				& minimal ratio: MILP/PC-check				&73.33	&22028	&$>$143006.89
\end{tabular}
}
\caption{Mean times in seconds to solve PCP with the MILP formulation over 25 and with PC-check over 50 instances.}
\label{tab: RT gap}
\end{table}

The two algorithms PC-check($\calC, \Gamma, \oplus, t$) and PC-check($\calC, \Gamma, \oplus, t, \calS, s$) show similar behavior of running times for different instance sizes. Figure~\ref{fig: rec plots} shows some of the running times for PC-check($\calC, \Gamma, \oplus, t$) and PC-check($\calC, \Gamma, \oplus, t, \calS, s$). Here, we can see that the running times increase with the number of evaluations $n$ and the number of statements $g$.

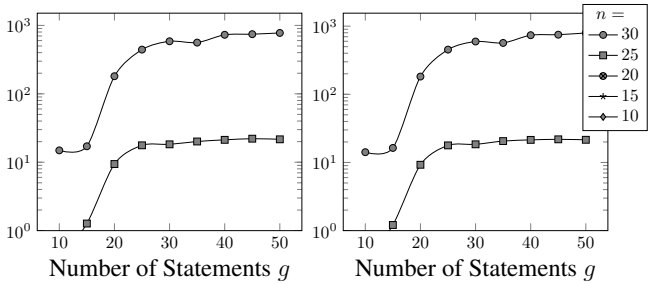
\begin{figure}[ht]
\resizebox{\columnwidth}{!}{
\begin{tikzpicture}
\begin{axis}[
cycle list name=black white,
smooth,
ymode = log,
  xlabel= \Large Number of Statements $g$,
  ymin = 1,
  legend style={ 
        at={(1.15,.6)}, 
        anchor= east
      }
  ]
\addplot coordinates{
(10,14.9)
(15,17.1)
(20,180.83)
(25,442.79)
(30,586.16)
(35,560.51)
(40,729)
(45,744.17)
(50,776)
};
\addplot coordinates{
(10,0.38)
(15,1.27)
(20,9.41)
(25,17.67)
(30,18.27)
(35,20.09)
(40,21.24)
(45,22.08)
(50,21.62)
};
\addplot coordinates{
(10,0.04)
(15,0.29)
(20,0.42)
(25,0.51)
(30,0.53)
(35,0.54)
(40,0.56)
(45,0.58)
(50,0.57)
};
\addplot coordinates{
(10,0.01)
(15,0.01)
(20,0.02)
(25,0.02)
(30,0.02)
(35,0.02)
(40,0.02)
(45,0.02)
(50,0.02)
};
\addplot coordinates{
(10,0.011)
(15,0.003)
(20,0.006)
(25,0.003)
(30,0.003)
(35,0.003)
(40,0.003)
(45,0.003)
(50,0.003)
};
\end{axis}
\end{tikzpicture}
\begin{tikzpicture}
\begin{axis}[
cycle list name=black white,
smooth,
ymode = log,
  xlabel= \Large Number of Statements $g$,
  ymin = 1,
  legend style={ 
        at={(1.15,.75)}, 
        anchor= east
      }
  ]
\addlegendimage{empty legend}
\addlegendentry{\hspace{-.8cm}\textbf{$n=$}}

\addlegendentry{$30$}
\addplot coordinates{
(10,14.16)
(15,16.34)
(20,182.09)
(25,452.74)
(30,595.09)
(35,567.16)
(40,736.87)
(45,749.54)
(50,795)
};
\addlegendentry{$25$}
\addplot coordinates{
(10,0.36)
(15,1.21)
(20,9.25)
(25,17.81)
(30,18.47)
(35,20.6)
(40,21.45)
(45,21.86)
(50,21.51)
};
\addlegendentry{$20$}
\addplot coordinates{
(10,0.03)
(15,0.28)
(20,0.42)
(25,0.51)
(30,0.54)
(35,0.54)
(40,0.57)
(45,0.57)
(50,0.57)
};
\addlegendentry{$15$}
\addplot coordinates{
(10,0.01)
(15,0.01)
(20,0.02)
(25,0.01)
(30,0.02)
(35,0.02)
(40,0.02)
(45,0.02)
(50,0.02)
};
\addlegendentry{$10$}
\addplot coordinates{
(10,0.003)
(15,0.001)
(20,0.003)
(25,0.001)
(30,0.001)
(35,0.001)
(40,0.001)
(45,0.001)
(50,0.001)
};
\end{axis}
\end{tikzpicture}
}
\caption{Mean times in seconds of PC-check($\calC, \Gamma, \oplus, t$) (left) and  PC-check($\calC, \Gamma, \oplus, t, \calS, s$) (right).}
\label{fig: rec plots}
\end{figure}

For time reasons experiments were not conducted for sizes $n \in \{40, 45, 50\}, g \in \{25, 30, \dots, 50\}$, and were only run for 5 instances (instead of 50) for some larger instance sizes. This may explain the irregularity at $n = 35 , g = 30$.
Table~\ref{tab: PC-check times} shows all running times for algorithms PC-check. 

\begin{table*}[t]
\centering
\resizebox{.78\textwidth}{!}{
\begin{tabular}{l|l|c|c|c|c|c|c|c|c|c}
$g=$ 				& 								& $n=10$	& $n=15$	& $n=20$	& $n=25$	& $n=30$	& $n=35$	& $n=40$	& $n=45$	& $n=50$	\\ \hline \hline
\multirow{3}{*}{10}	& PC-check($\calC, \Gamma, \oplus, t$)			&0.011	&0.01		&0.04		&0.38		&14.9		&0.005	&0.004	&0.004 	&0.004	\\ \cline{3-11}
				& PC-check($\calC, \Gamma, \oplus, t, \calS, s$)	&0.003	&0.01		&0.03		&0.36		&14.16	&0.003	&0.003	&0.004	&0.003	\\ \cline{3-11}
				& ratio							&3.67		&1		&1.33		&1.06		&1.05		&1.67		&1.33		&1		&1.33		\\ \hline 
\multirow{3}{*}{15}	& PC-check($\calC, \Gamma, \oplus, t$)			&0.003	&0.01		&0.29		&1.27		&17.1		&522.26	&1.06		&0.06		&7.93		\\ \cline{3-11}
				& PC-check($\calC, \Gamma, \oplus, t, \calS, s$)	&0.001	&0.01		&0.28		&1.21		&16.34	&500.13	&0.97		&0.05		&7.26		\\ \cline{3-11}
				& ratio							&3		&1		&1.04		&1.05		&1.05		&1.04		&1.09		&1.2		&1.09		\\ \hline
\multirow{3}{*}{20}	& PC-check($\calC, \Gamma, \oplus, t$)			&0.006	&0.02		&0.42		&9.41		&180.83	&2165.83	&$>$9072	&$>$7344	&2010.45	\\ \cline{3-11}
				& PC-check($\calC, \Gamma, \oplus, t, \calS, s$)	&0.003	&0.02		&0.42		&9.25		&182.09	&2159.93	&$>$9072	&$>$7344	&2025.72	\\ \cline{3-11}
				& ratio							&2		&1		&1		&1.02		&0.99		&1.00		&		&		&0.99		\\ \hline
\multirow{3}{*}{25}	& PC-check($\calC, \Gamma, \oplus, t$)			&0.003	&0.02		&0.51		&17.67	&442.79	&5393.47*	&		&		&		\\ \cline{3-11}
				& PC-check($\calC, \Gamma, \oplus, t, \calS, s$)	&0.001	&0.01		&0.51		&17.81	&452.74	&5377.27*	&		&		&		\\ \cline{3-11}
				& ratio							&3		&2		&1		&0.99		&0.98		&1.00		&		&		&		\\ \hline
\multirow{3}{*}{30}	& PC-check($\calC, \Gamma, \oplus, t$)			&0.003	&0.02		&0.53		&18.27	&586.16	&6.57*	&		&		&		\\ \cline{3-11}
				& PC-check($\calC, \Gamma, \oplus, t, \calS, s$)	&0.001	&0.02		&0.54		&18.47	&595.09	&6.32*	&		&		&		\\ \cline{3-11}
				& ratio							&3		&1		&0.98		&0.99		&0.98		&1.04		&		&		&		\\ \hline
\multirow{3}{*}{35}	& PC-check($\calC, \Gamma, \oplus, t$)			&0.003	&0.02		&0.54		&20.09	&560.51	&16796.17*	&		&		&		\\ \cline{3-11}
				& PC-check($\calC, \Gamma, \oplus, t, \calS, s$)	&0.001	&0.02		&0.54		&20.6		&567.16	&16503.98*	&		&		&		\\ \cline{3-11}
				& ratio							&3		&1		&1		&0.98		&0.99		&1.02		&		&		&		\\ \hline
\multirow{3}{*}{40}	& PC-check($\calC, \Gamma, \oplus, t$)			&0.003	&0.02		&0.56		&21.24	&729		&24494.72*	&		&		&		\\ \cline{3-11}
				& PC-check($\calC, \Gamma, \oplus, t, \calS, s$)	&0.001	&0.02		&0.57		&21.45	&736.87	&24269.56*	&		&		&		\\ \cline{3-11}
				& ratio							&3		&1		&0.98		&0.99		&0.99		&1.01		&		&		&		\\ \hline
\multirow{3}{*}{45}	& PC-check($\calC, \Gamma, \oplus, t$)			&0.003	&0.02		&0.58		&22.08	&744.17	&23180.91*	&		&		&		\\ \cline{3-11}
				& PC-check($\calC, \Gamma, \oplus, t, \calS, s$)	&0.001	&0.02		&0.57		&21.86	&749.54	&22886.84*	&		&		&		\\ \cline{3-11}
				& ratio							&3		&1		&1.02		&1.01		&0.99		&1.01		&		&		&		\\ \hline
\multirow{3}{*}{50}	& PC-check($\calC, \Gamma, \oplus, t$)			&0.003	&0.02		&0.57		&21.62	&776		&27084.24*	&		&		&		\\ \cline{3-11}
				& PC-check($\calC, \Gamma, \oplus, t, \calS, s$)	&0.001	&0.02		&0.57		&21.51	&795		&26955.88*	&		&		&		\\ \cline{3-11}
				& ratio							&3		&1		&1		&1.01		&0.98		&1.00		&		&		&		
\end{tabular}
}
\caption{Mean times of PC-check in seconds fixing $m=25$ and ratios of the mean times between PC-check($\calC, \Gamma, \oplus, t$) and PC-check($\calC, \Gamma, \oplus, t, \calS, s$) rounded to the nearest hundredth. *Mean time over 5 instances only. (All remaining are mean times over 50 instances.)}
\label{tab: PC-check times}
\end{table*}

A detailed analysis shows that in 57\% of the measured instance sizes  PC-check($\calC, \Gamma, \oplus, t$) is slower than PC-check($\calC, \Gamma, \oplus, t, \calS, s$). 
The ratios of the mean running times of the two algorithms demonstrate that PC-check($\calC, \Gamma, \oplus, t$) is at most 3.67 times slower than PC-check($\calC, \Gamma, \oplus, t, \calS, s$) and PC-check($\calC, \Gamma, \oplus, t, \calS, s$) is at most 1.03 times slower than PC-check($\calC, \Gamma, \oplus, t$). 
This insignificant difference between the running times of the two algorithms means that for these instances reducing the number of backtracks by ruling out conflicting sets is not worth the effort and space needed to maintain the list of conflicting sets (of size $<5$). 

To find an explanation for the behavior of the running times as shown in Figure~\ref{fig: rec plots},  
we observed the occurrence of instances that have solutions in $\calC(1)$, in $\calC(t)$ with $t>1$, or are inconsistent.
The whole search space must be explored until deciding inconsistency for all $t$, which can lead to high running times. PC-check solves $\calC(1)$ instances in polynomial time. The instance distribution may suggest that the running times go up with the number of inconsistent instances, e.g., see Figure~\ref{fig: percentages}.

\begin{figure}[ht]
\resizebox{\columnwidth}{!}{
\begin{tikzpicture}
\begin{axis}[
  height=4.8cm,
  width=12cm,
  ybar,
  enlargelimits=0.02,
  xlabel={Number of Statements},
  ylabel={Percentage},
  extra x ticks = {5, 10, 15, 20, 25, 30, 35, 40, 45, 50, 55},
  extra y ticks = {20, 40, 60,  80, 100, 110},
  legend style={at={(1,0.99)},anchor=south east,legend columns=-1},
  ybar interval=0.6,
  cycle list = {black,black!70,black!40,black!10}
   ] 
\addplot+[] coordinates{
(10, 2.0)(15, 2.0)(20, 24.0)(25, 58.0)(30, 76.0)(35, 72.0)(40, 92.0)(45, 94.0)(50, 100.0)(55,0)
};
\addplot+[fill,text=black] coordinates{
(10, 78.0)(15, 94.0)(20, 76.0)(25, 42.0)(30, 24.0)(35, 28.0)(40, 8.0)(45, 6.0)(50, 0.0)(55,0)
};
\addplot+[fill,text=black] coordinates{
(10, 20.0)(15, 4.0)(20, 0.0)(25, 0.0)(30, 0.0)(35, 0.0)(40, 0.0)(45, 0.0)(50, 0.0)(55,0)
};
\legend{{Inconsistent for all $t$, },{$C(t)$-consistent, $t>1$, }, $C(1)$-consistent}
\end{axis}
\end{tikzpicture}
}
\caption{Percentages of instance classes for the tested random instances with $n=30$ (top) and $n=35$ (bottom).}
\label{fig: percentages}
\end{figure}
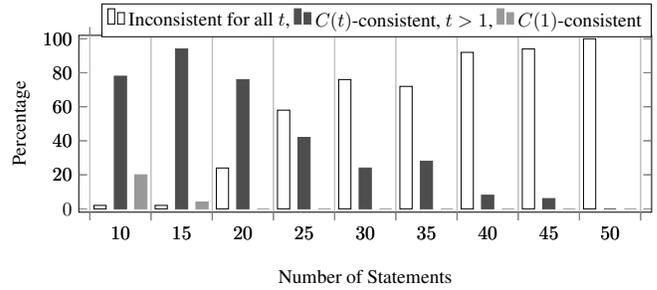

\section{Conclusion}
Exploiting the theoretical results on properties of consistent instances developed in Section~\ref{sec: rec algo} allow the algorithms PC-check to prune the search space much further than a MILP solver could do for the MILP formulation given in Section~\ref{sec: MILP}. The experimental results confirm, that the algorithms PC-check are solving the instances faster than CPLEX. Even more, the ratios between the mean solving times of the MILP and PC-check increase extremely quickly with the number of evaluations and statements. 
It is not obvious how the pruning rules of the PC-check algorithms can be incorporated in the MILP formulation as constraints.

There is no significant difference between the mean running times of the two recursive algorithms PC-check on the tested instances. 
Thus, in PC-check($\calC, \Gamma, \oplus, t, \calS=\emptyset, s$), the effort of maintaining a list $\calS$ of (possibly exponentially many) conflicting sets to prune the search space further, is not paying off. An analysis of the instance consistency types indicates that the running times for algorithms PC-check might increase with the number of inconsistent instances. 

A further analysis could involve the size of the search space, i.e., counting the number of $\Gammadestrict$-satisfying HCLP models and the number of HCLP models that were actually considered during the search. In future work, we will try using a relaxation of a MILP formulation as a fast check for inconsistency within PC-check($\calC, \Gamma, \oplus, t$). If the relaxation shows that the current subproblem is inconsistent, we can avoid another (time consuming) recursive call.

	\section*{Acknowledgment}
	This publication has emanated from research conducted
	with the financial support of Science Foundation Ireland
	(SFI) under Grant Number SFI/12/RC/2289. 

\bibliography{reference}

\begin{thebibliography}{10}

\bibitem{BoothCLMS10}
R.~Booth, Y.~Chevaleyre, J.~Lang, J.~Mengin, and C.~Sombattheera.
\newblock Learning conditionally lexicographic preference relations.
\newblock In {\em {ECAI}~'10, Lisbon, Portugal, August 16-20, 2010, Proceedings}, pages 269--274, 2010.

\bibitem{BoutilierBDHP04}
C.~Boutilier, R.~I. Brafman, C.~Domshlak, H.~Hoos, and D.~Poole.
\newblock {CP}-nets: A tool for reasoning with conditional \textit{ceteris paribus} preference statements.
\newblock {\em Journal of Artificial Intelligence Research}, 21:135--191, 2004.

\bibitem{Huellermeier-PL-12-Learn-Lex-Trees}
M.~Br\"auning and E.~H\"ullermeier.
\newblock Learning conditional lexicographic preference trees.
\newblock In {\em Preference Learning (PL-12), ECAI~'12 workshop, Montpellier, France, August 28th, 2012, Proceedings}, 2012.

\bibitem{BridgRicci07}
D.~Bridge and F.~Ricci.
\newblock Supporting product selection with query editing recommendations.
\newblock In {\em RecSys~'07, Minneapolis, MN, USA, October 19-20, 2007, Proceedings}, pages 65--72. ACM, 2007.

\bibitem{Dombi-Learning-Lex-EJOR-2007}
J.~Dombi, C.~Imreh, and N.~Vincze.
\newblock Learning lexicographic orders.
\newblock {\em European Journal of Operational Research}, 183(2):748--756, 2007.

\bibitem{MCDA2005}
J.~Figueira, S.~Greco, and M.~Ehrgott.
\newblock {\em Multiple Criteria Decision Analysis---State of the Art Surveys}.
\newblock Springer International Series in Operations Research and Management Science Volume 76, 2005.

\bibitem{FlachM07-Lex-ranker}
P.~A. Flach and E.~Matsubara.
\newblock A simple lexicographic ranker and probability estimator.
\newblock In {\em ECML~'07, Warsaw, Poland, September 17-21, 2007, Proceedings}, pages 575--582, 2007.

\bibitem{Fuern-Hueller-Preference-Learning}
J.~F\"urnkranz and E.~H\"ullermeier.
\newblock {\em Preference Learning}.
\newblock Springer-Verlag New York, Inc., New York, NY, USA, 1st edition, 2010.

\bibitem{GeorgeRW15}
A.-M. George, A.~Razak, and N.~Wilson.
\newblock The comparison of multi-objective preference inference based on lexicographic and weighted average models.
\newblock In {\em ICTAI~'15, Vietri sul Mare, Italy, November 9-11, 2015, Proceedings}, pages 88--95, 2015.

\bibitem{GeorgeW16}
A.-M. George, N.~Wilson, and B.~O'Sullivan.
\newblock Towards fast algorithms for the preference consistency problem based on hierarchical models.
\newblock In {\em IJCAI~'16, New York, USA, July 09-15, 2016, Proceedings.}, 2016.

\bibitem{LiuT15}
X.~Liu and M.~Truszczynski.
\newblock Learning partial lexicographic preference trees over combinatorial domains.
\newblock In {\em AAAI~'15, Austin Texas, USA, January 25–30, 2015, Proceedings}, pages 1539--1545. AAAI Press, 2015.

\bibitem{marinescu2013moopt}
R.~Marinescu, A.~Razak, and N.~Wilson.
\newblock Multi-objective constraint optimization with tradeoffs.
\newblock In {\em CP~'13, Uppsala, Sweden, September 16-20, 2013, Proceedings}, pages 497--512, Berlin, Heidelberg, 2013. Springer Berlin Heidelberg.

\bibitem{Trabelsi11}
W.~Trabelsi, N.~Wilson, D.~Bridge, and F.~Ricci.
\newblock Preference dominance reasoning for conversational recommender systems: a comparison between a comparative preferences and a sum of weights approach.
\newblock {\em International Journal on Artificial Intelligence Tools}, 20(4):591--616, 2011.

\bibitem{WilsonBorning1993}
M.~Wilson and A.~Borning.
\newblock Hierarchical constraint logic programming.
\newblock {\em The Journal of Logic Programming}, 16(3-4):277--318, 1993.

\bibitem{WilsonECAI14}
N.~Wilson.
\newblock Preference inference based on lexicographic models.
\newblock In {\em ECAI~'14, 18-22 August 2014, Prague, Czech Republic - Including PAIS~'14}, pages 921--926, 2014.

\bibitem{WilsonGOSIJCAI2015longer}
N.~Wilson, A.-M. George, and B.~O'Sullivan.
\newblock {\em Computation and Complexity of Preference Inference Based on Hierarchical Models (extended version of IJCAI-2015)}.
\newblock http://ucc.insight-centre.org/nwilson/PrefInfHCLPproofs.pdf), 2015.

\end{thebibliography}

\end{document}